\documentclass{article}

\usepackage[margin = 1in]{geometry}
\usepackage{algorithm}
\usepackage{algpseudocode}
\usepackage{xspace}
\usepackage[inline]{enumitem}
\usepackage{hyperref}
\usepackage{amsmath, amssymb, amsthm}
\usepackage{cleveref}
\usepackage{xcolor}
\usepackage{thm-restate}
\usepackage{multicol}

\newtheorem{assumption}{Assumption}
\newtheorem{theorem}{Theorem}
\newtheorem{lemma}[theorem]{Lemma}
\newtheorem{definition}[theorem]{Definition}
\newtheorem{corollary}[theorem]{Corollary}
\newtheorem{claim}[theorem]{Claim}

\newcommand{\abs}[1]{\left\lvert #1 \right\rvert}
\newcommand{\braces}[1]{\left\{#1\right\}}
\newcommand{\bracket}[1]{\left[#1\right]}
\newcommand{\ceil}[1]{\left\lceil #1 \right\rceil}

\newcommand{\paren}[1]{\left(#1\right)}
\newcommand{\set}[1]{\braces{#1}}

\newcommand{\NN}{\ensuremath{\mathbb{N}}}
\newcommand{\RR}{\ensuremath{\mathbb{R}}}
\newcommand{\ZZ}{\ensuremath{\mathbb{Z}}}
\newcommand{\half}{\ensuremath{\frac 1 2}}
\newcommand{\B}{\mathcal{B}}

\newcommand{\conlamufl}{\textsf{Cont-$\lambda$-UFL}\xspace}
\newcommand{\contkmed}{\textsf{Cont-$k$-Med}\xspace}
\newcommand{\contkmeans}{\textsf{Cont-$k$-Means}\xspace}
\newcommand{\contkp}{\textsf{Cont-$(k,p)$-Clustering}\xspace}

\newcommand{\contkcwo}{\textsf{Cont-$k$CwO}\xspace}
\newcommand{\kcwo}{{\sf $k$-CwO}\xspace}
\newcommand{\contfairkmed}{\textsf{ContFair-$k$-Med}\xspace}

\newcommand{\child}{\mathsf{child}}
\newcommand{\cost}{{\mathsf{cost}}}
\newcommand{\costf}{{\mathsf{cost_F}}}
\newcommand{\costc}{{\mathsf{cost_C}}}
\newcommand{\cov}{\ensuremath{\mathrm{cov}}}
\newcommand{\dee}{\mathrm{d}}
\newcommand{\DLP}{{\sf DLP}}
\newcommand{\even}{{\sf even}}
\newcommand\eps{\ensuremath{\varepsilon}}
\newcommand\Exp{\mathbf{Exp}}
\newcommand{\Ie}{I_\eps}
\newcommand{\lp}{\textsf{LP}\xspace}

\newcommand{\NP}{{\sf NP}}
\newcommand{\odd}{{\sf odd}}
\newcommand{\opt}{{\tt opt}}
\renewcommand{\P}{{\sf P}}
\newcommand{\ra}{r_\alpha}

\newcommand{\Reps}{\mathsf{Reps}}
\newcommand{\Repscov}{\mathsf{Reps}_{\mathrm{cov}}}
\newcommand{\Repsa}{\mathsf{Reps_\alpha}}
\newcommand{\RepsCy}{\mathsf{Reps}_{(C,y)}}
\newcommand{\RepsCCy}{\mathsf{Reps}_{(C^2,y)}}
\newcommand{\RepsCya}{\mathsf{Reps_{(C,y),\alpha}}}
\newcommand{\treeroot}{{\rm root}}
\newcommand{\Scov}{S_{\mathrm{cov}}}
\newcommand{\Sep}{{\sf Sep}}

\newcommand{\SepCy}{\mathsf{Sep}_{(C,y)}}
\newcommand{\SepCCy}{\mathsf{Sep}_{(C^2,y)}}
\newcommand{\SepCya}{\mathsf{Sep}_{(C,y),\alpha}}
\newcommand{\Sepcov}{\mathsf{Sep}_{\mathrm{cov}}}

\title{Approximation Algorithms for Continuous Clustering and Facility Location Problems}
\author{}
\date{}
\begin{document}

\clearpage
\begin{center}
{\Large Approximation Algorithms for Continuous \\ \vspace{0.5em}Clustering and Facility Location Problems}

\vspace{2em}

\begin{tabular}{ccc}
{\large Deeparnab Chakrabarty} & {\large Maryam Negahbani} & {\large Ankita Sarkar}\\
\href{mailto:deeparnab@dartmouth.edu}{\tt deeparnab@dartmouth.edu} & \href{mailto:maryam.negahbani.gr@dartmouth.edu}{\tt maryam.negahbani.gr@dartmouth.edu} & \href{mailto:ankita.sarkar.gr@dartmouth.edu}{\tt ankita.sarkar.gr@dartmouth.edu}\\
Department of Computer Science & Department of Computer Science & Department of Computer Science\\
Dartmouth College, NH, USA & Dartmouth College, NH, USA & Dartmouth College, NH, USA
\end{tabular}
\begin{abstract}
In this paper, we consider center-based clustering problems where $C$, the set of points to be clustered, lies in a metric space $(X,d)$, and the set $X$ of candidate centers is potentially infinite-sized. We call such problems {\em continuous} clustering problems to differentiate them from the {\em discrete} clustering problems where the set of candidate centers is explicitly given.
It is known that for many objectives, when one restricts the set of centers to $C$ itself and applies
an $\alpha_\mathsf{dis}$-approximation algorithm for the discrete version, one obtains a $\beta \cdot \alpha_{\mathsf{dis}}$-approximation algorithm for the continuous version via the triangle inequality property of the distance function. Here $\beta$ depends on the objective, and for many objectives such as $k$-median, $\beta = 2$, while for some others such as $k$-means, $\beta = 4$.
The motivating question in this paper is \emph{whether this gap of factor $\beta$ between continuous and discrete problems is inherent, or can one design better algorithms for continuous clustering than simply reducing to the discrete case as mentioned above?}
In a recent SODA 2021 paper, Cohen-Addad, Karthik, and Lee prove a
factor-$2$ and a factor-$4$ hardness, respectively, for the continuous versions of the $k$-median and $k$-means problems, even when the number of cluster centers is a constant.
The discrete problem for a constant number of centers is easily solvable exactly using enumeration, and therefore, in certain regimes, the ``$\beta$-factor loss'' seems unavoidable.

In this paper, we describe a technique based on the {\em round-or-cut} framework to approach continuous clustering problems. We show that, for the continuous versions of some clustering problems, we can design
approximation algorithms attaining a better factor than the $\beta$-factor blow-up mentioned above. In particular, we do so for: the uncapacitated facility location problem with uniform facility opening costs $(\lambda$-UFL); the $k$-means problem; the individually fair $k$-median problem; and the $k$-center with outliers problem. Notably, for $\lambda$-UFL, where $\beta = 2$ and the discrete version is NP-hard to approximate within a factor of $1.27$, we describe a $2.32$-approximation for the continuous version, and indeed $2.32 < 2 \times 1.27$. Also, for $k$-means, where $\beta = 4$ and the best known approximation factor for the discrete version is $9$, we obtain a $32$-approximation for the continuous version, which is better than $4 \times 9 = 36$.

The main challenge one faces  is that most algorithms for the discrete clustering problems, including the state of the art solutions, depend on Linear Program (LP) relaxations that become infinite-sized in the continuous version. To overcome this, we design new linear program relaxations for the continuous clustering problems which,  although having exponentially many constraints,  are amenable to the round-or-cut framework.

\end{abstract}

\begin{minipage}{0.89\textwidth}
{\small~\\

\textbf{Related version.} \textit{ESA 2022 version:} \url{https://drops.dagstuhl.de/opus/volltexte/2022/16971}~\cite{ESA-version}\\

\textbf{Funding.}\quad The authors were supported by NSF award \#2041920.}
\end{minipage}
\end{center}

\section{Introduction}
Clustering is a ubiquitous problem arising in various areas ranging from data analysis to operations research. 
One popular class of clustering problems are the so-called {\em center-based clustering} problems where the quality of the clustering is determined by a  function of the distances of every point in $C$ to the ``centers'' of the clusters they reside in. Two extensively studied measures are the sum of these distances, with the resulting problem called the $k$-median problem, and the sum of squares of these distances, with the resulting problem called the $k$-means problem.

In most settings, these center-based clustering problems are \NP-hard and one considers approximation algorithms for the same. Traditionally, however, approximation algorithms for these problems have been studied in {\em finite/discrete} metric spaces and, in fact, usually under the constraint that the set of centers, aka facilities, can be selected only from a prescribed subset $F\subseteq X$ of the metric space. Indeed, this model makes perfect sense when considering applications in operations research where the possible depot-locations may be constrained. These discrete problems have been extensively studied~\cite{HochbS1985,shmoys-tardos-aardal-1997,cgst-primal-rounding,guha-khuller-hardness,ChariKMN2001,CharG1999,JainV2001,JainMMSV2003,AryaGKMMP2001,KanunMNPSW2004,Li2013,LiS2016,ByrkaGRS2013} over the last three decades. 
For instance, for the $k$-median problem, the best known approximation algorithm is a $2.675$-approximation \cite{k-median-approx}, while the best known hardness is $(1+2/e) \approx 1.74$~\cite{JainMMSV2003}.
For the $k$-means problem, the best known approximation algorithm is a $(9+\eps)$-approximation algorithm~\cite{Ahmadian2017BetterGF,GuptaT2008,KanunMNPSW2004}, while the best hardness is $1 + 8/e \approx 3.94$~\cite{JainMMSV2003}.

Restricting to a finite metric space, however, makes the problem easier, and indeed many of the above algorithms in the papers mentioned above would be infeasible to implement if $X$ were extremely large -- for instance, if $X$ were $\RR^m$ for some large dimension $m$, and the distance function were the $\ell_p$-metric, for some $p$.
On the other hand, it is reasonably easy to show using triangle inequality that if one considers opening centers from $C$ itself, and thus reduces the problem to its discrete version, then one incurs a hit of a factor $\beta$ in the approximation factor, where $\beta$ is a constant depending on the objective function. In particular, if we look at the sum-of-distances objectives such as in $k$-median, then $\beta = 2$, while if one looks at the sum-of-squared-distances such as in $k$-means, then $\beta = 4$.
Therefore, one immediately gets a $5.35$-approximation for the {\em continuous $k$-median} problem and a $36+\eps$-approximation for the {\em continuous $k$-means problem}.  The question we investigate  is 
\begin{quote}\emph{
	Is this factor $\beta$ hit necessary between the continuous and discrete versions of center based clustering problems, or can one design better approximation algorithms for the continuous case?
}
\end{quote}
It is crucial to note that when considering designing algorithms, we do not wish to make {\em any} assumptions on the underlying metric space $(X,d)$. For instance, we do not wish to assume $X = \RR^m$ for some $m$.
This is important, for we really want to compare ourselves with the $\beta$ which is obtained using only the triangle-inequality and symmetry property of $d$. On the other hand, to exhibit that a certain algorithm does not work, any candidate metric space suffices.

Recently, in a thought-provoking paper~\cite{cohen-addad-et-al}, Cohen-Addad, Karthik, and Lee show that, unless $\P = \NP$, the $k$-median and $k$-means problem defined on $(\RR^m, \ell_\infty)$ cannot have an approximation ratio better than $2$ and $4$, respectively, {\em even when} $k$ is a constant! Since the discrete problems have trivial exact algorithms via enumeration when $k$ is a constant, this
seems to indicate that in certain cases the above factor $\beta$ hit is unavoidable. Is it possible that the {\em inapproximability} of the continuous problem is  indeed $\beta$ times the inapproximability of the discrete version? 

\subsection{Our Results}
Our main contribution is a direct approach towards the continuous versions of clustering problems. We apply this to the following
clustering problems where we obtain a factor better than $\beta \cdot \alpha_{\mathsf{dis}}$, where $\alpha_{\mathsf{dis}}$ is the best known factor for the discrete version of the problem.

\begin{itemize}
	\item In the continuous $\lambda$-UFL problem, a ``soft'' version of the continuous $k$-median problem, one is allowed to pick any number of centers but has to pay a parameter $\lambda$ 
	for each picked center. The objective is to minimize the sum of distances of points in $C$ to picked centers plus the cost for opening these centers.
	Again, note that the centers can be opened anywhere in $X$. For the discrete version, where the only possible center locations are in $C$, there is a $1.488$-approximation due to Li~\cite{Li2013}, and a hardness of approximation within a factor of $1.278$ is known due to Guha and Khuller~\cite{guha-khuller-hardness}. We describe a $2.32$-approximation algorithm. Note that $2.32 < 2\cdot 1.278$ and thus, for this problem, the inapproximability is {\em not} $\beta$ times that of the discrete case. We also show how the reduction of~\cite{cohen-addad-et-al} carries over, to prove a hardness of $2$ for this problem.
	
	\item In the continuous $k$-means problem, we wish to minimize the sum of squares of distances of clients to the closest open center. 
	 Recall that for this problem we have $\beta = 4$, and thus one gets a 
	 $36$-factor algorithm for the continuous $k$-means using the best known $9$-factor~\cite{Ahmadian2017BetterGF,GuptaT2008,KanunMNPSW2004} algorithm for the discrete problem. 
	We describe an improved $32$-approximation for the continuous $k$-means problem.
	
	\item For the continuous $k$-median problem, our techniques fall short of improving $2\times$ the best known approximation factor for the discrete $k$-median problem. On the other hand, we obtain better algorithms for the
	the {\em individually fair} or {\em priority} version of the continuous $k$-median problem. In this problem, every point $v\in C$ has a specified radius $r(v)$ and desires a center opened within this distance.
	The objective is the same as the $k$-median problem: minimize the sum of all the distances.
	This problem arises as a possible model~\cite{JungKL2020,chakrabarty-negahbani-fairness,MahabV2020,VakilY2021, Ples87,BCCN21} in the study of fair clustering solutions, since the usual $k$-median algorithms may place certain clients inordinately far away. 
	At a technical level, this problem is a meld of the $k$-median and the $k$-center problems; the latter is $\NP$-hard, which forces one to look at {\em bicriterion} approximations. An $(\alpha,\gamma)$-approximation would return a solution within $\alpha$ times the optimum but may connect $v$ to a point as far as $\gamma r(v)$ away. Again, any $(\alpha, \gamma)$-approximation for the {\em discrete} version where $X=C$ would imply a $(2\alpha, 2\gamma)$-approximation for the continuous version.
	
	The best discrete approximation is $(8,3)$ due to Vakilian and Yal{\c{c}}{\i}ner~\cite{VakilY2021} which would imply an $(16,3)$-approximation for the continuous version. We describe an $(8,8)$-approximation for the continuous version of the problem.

	\item In the $k$-center with outliers (kCwO) problem, we are given a parameter $m \leq |C|$, and we need to serve only $m$ of the clients. The objective is the maximum distance of a served client to its center. 
	The $k$-center objective is one of the objectives for which most existing discrete algorithms can compare themselves directly with the continuous optimum. The $3$-approximation algorithm in~\cite{ChariKMN2001} for the kCwO problem is one such example. However, the 
	best known algorithm for kCwO for the discrete case (when $X=C$)
	is a $2$-approximation by Chakrabarty, Goyal, and Krishnaswamy~\cite{ChakrGK2020} which proceeds via LP rounding, and does not give a $2$-approximation for continuous kCwO. This was explicitly noted in a work by Ding, Yu, and Wang (\emph{``... unclear of the resulting approximation ratio for the problem in Euclidean space.''})~\cite{DingYW19}, that describes a $2$-approximation for kCwO in Euclidean space, however, violating the number of clients served. We give a proper $2$-approximation for the continuous kCwO problem (with no assumptions on the metric space) with no violations.
\end{itemize}

\subsection{Our Technical Insight}
Most state of the art approximation algorithms for center-based clustering problems are based on LP relaxations where one typically has variables $y_i$ for {\em every} potential location 
of a center. When the set $X$ is large, this approach becomes infeasible. Our main technical insight, underlying all our results, is to use a different style of linear program with polynomially many variables but exponentially many constraints. We then use the {\em round-or-cut} framework to obtain our approximation factor. More precisely, given a {\em potential} solution to our program, we either ``round'' it to get a desired solution within the asserted approximation factor, or we find a {\em separating hyperplane} proving that this potential solution is infeasible.
Once this hyperplane is fed to the ellipsoid algorithm~\cite{ellipsoid}, the latter generates another potential solution, and the process continues. Due to the ellipsoid method's guarantees, we obtain our approximation factor in polynomial time. 

For every client $v\in C$, our LP relaxation has variables of the form $y(v,r)$, indicating whether there is some point $x\in X$ in an $r$-radius around $v$ which is ``open'' as a center. Throughout the paper we use $r$ as a quantity varying ``continuously'', but it can easily be discretized, with a loss of at most $\frac{1}{\mathrm{poly}(n)}$, to arise from a set of size $\leq \mathrm{poly}(n)$. Thus there are only polynomially many such variables. We add the natural ``monotonicity'' constraints: $y(v,r) \leq y(v,s)$ whenever $r\leq s$. Interestingly, for one of the applications, we also need the monotonicity constraints for non-concentric balls: if $B(v,r) \subseteq B(u,s)$, then we need $y(v,r)\leq y(u,s)$.

We have a variable $C_v$ indicating the cost the client $v$ pays towards the optimal solution. 
Next, we connect the $C_v$'s and the $y(v,r)$'s in the following ways (when $\beta = 1$, and something similar when $\beta = 2$).
One connection states that for any $r$, $C_v \geq r\cdot (1 - y(v,r))$ and we add these to our LP. For the last two applications listed above, this suffices.
However, one can also state the stronger condition of $C_v \geq \int_{0}^\infty (1 - y(v,r))dr$. Indeed, the weaker constraint is the ``Markov-style inequality''
version of the stronger constraint.

Our second set of constraints restrict the $y(v,r)$'s to be ``not too large''. For instance, for the fair  $k$-median or kCwO problems where we are only allowed $k$ points from $X$, we assert that for {\em any} set $\mathcal{B}$ of disjoint balls $B(v,r)$, we must have the sum of the respective $y(v,r)$'s to be at most $k$. This set of constraints is exponentially many, and this is the set of constraints that need the round-or-cut machinery. 
For the $\lambda$-UFL problem, we have that the sum of the $y(v,r)$'s scaled by $\lambda$ plus the sum of the $C_v$'s should be at most $\opt_g$, which is a running guess of $\opt$.

Once we set up the framework above, then we can {\em port many existing rounding} algorithms for the discrete clustering problems without much hassle.
In particular, this is true for rounding algorithms which use the $C_v$'s as the core driving force. For the continuous $\lambda$-UFL problem, we port the rounding algorithm from the paper~\cite{shmoys-tardos-aardal-1997} by Shmoys, Tardos, and Aardal. 
For the continuous $k$-means problem, we port the rounding algorithm from the paper~\cite{cgst-primal-rounding} by Charikar, Guha, Shmoys, and Tardos.
For the continuous fair $k$-median problem, we port the rounding algorithm from the paper~\cite{chakrabarty-negahbani-fairness} by Chakrabarty and Negahbani, which itself builds on the algorithm present in the paper~\cite{AlamdS2017} by Alamdari and Shmoys. For the continuous kCwO problem, we port the rounding algorithm present in the paper~\cite{ChakrGK2020} by Chakrabarty, Goyal, and Krishnaswamy.

Our results fall short for the continuous $k$-median problem (without fairness), where we can port the rounding algorithm from the paper~\cite{cgst-primal-rounding} 
and get a $6.67$-approximation. This, however, does not improve upon the $5.35$-factor mentioned earlier.

\subsection{Other Related Works and Discussion}
The continuous $k$-means and median problems have been investigated quite a bit in the specific setting when $X = \RR^m$ and when $d(u,v)$ is the $\ell_2$ distance.
The paper~\cite{Matousek2000} by Matou\v{s}ek describes an $(1+\eps)$-approximation (PTAS) that runs in time $O(n\log^k n \cdot \eps^{-O(k^2m)})$. This led to a flurry of results~\cite{HarPM2004,delaVFKKR03,KumarSS2004,Chen2006,FeldmanMS2007} on obtaining PTASes with better dependencies on $k$ and $m$ via the applications of coresets. There is a huge and growing literature on coresets, and we refer the interested reader to the paper~\cite{CohenSS2021} by Cohen-Addad, Saulpic, and Schwiegelshohn, and the references within, for more information.
Another approach to the continuous $k$-means problem has been local search.
The paper~\cite{KanunMNPSW2004} which describes a $9+\eps$-approximation was first stated for the geometric setting, however it also went via the discretization due to Matou\v{s}ek~\cite{Matousek2000} and suffered a running time of exponential dependency on the dimension. More recent papers~\cite{FriggstadRS2019l,CohenKM2019l} described local-search based PTASes for metrics with doubling dimension $D$, with running time exponentially depending on $D$. These doubling metrics generalize $(\RR^m, \ell_2)$-metrics.
However, none of the above ideas seem to suggest better constant factor approximations for the continuous $k$-median/means problem in the general case, and indeed even when $X = \RR^m$ but $m$ is part of the input.

The $k$-means problem in the metric space $(\RR^m, \ell_2)$, where $m$ and $k$ are not constants, has been studied extensively~\cite{Trevisan2000,AwasthiCKS2015,Cohen-AddadK2019,Cohen-AddadKL2022,Ahmadian2017BetterGF,Cohen-AddadEMN2022}, and is called the Euclidean $k$-means problem. The discrete version of this problem was proved {\sf APX}-hard in 2000~\cite{Trevisan2000}, but the {\sf APX}-hardness of the continuous version was proved much later, in 2015~\cite{AwasthiCKS2015}. More recently, the hardness results for both versions have been improved: the discrete Euclidean $k$-means problem is hard to approximate to factor $1.17$, while the continuous problem is hard to approximate to factor $1.07$~\cite{Cohen-AddadK2019}. Moreover, under assumption of a complexity theoretic hypothesis called the Johnson coverage hypothesis, these numbers have been improved to $1.73$ and $1.36$, respectively~\cite{Cohen-AddadKL2022}. On the algorithmic side, the discrete Euclidean $k$-means problem admits a better approximation ratio than the general case: a $6.36$ approximation was described in 2017~\cite{Ahmadian2017BetterGF}, which was very recently improved to $5.912$~\cite{Cohen-AddadEMN2022}.

We believe that our paper takes the first stab at getting approximation ratios better than $\beta \times$ the best discrete factor for the {\em continuous} clustering problems.
Round-or-cut is a versatile framework for approximation algorithm design with many recent applications~\cite{CarrFKP2001,ChakrCKK2015,AnSS2017,chakrabarty-negahbani-f-center,AneggAKZ22}, and the results in our paper is yet another application of this paradigm. However, many questions remain.
We believe that the most interesting question to tackle is the continuous $k$-median problem. The best known discrete $k$-median algorithms are, in fact, combinatorial in nature, and are obtained via 
 applying the primal-dual/dual-fitting based methods~\cite{JainV2001,JainMMSV2003,LiS2016,k-median-approx} on the discrete LP. However, their application still needs an explicit description of the facility set, and it is interesting to see if they can be {\em directly} ported to the continuous setting.

All the algorithms in our paper, actually {\em still} open centers from $C$. Even then, we are able to do better than simply reducing to the discrete case, because we do not commit to the $\beta$ loss upfront, and instead round from a fractional solution that can open centers anywhere in $X$. 
This raises an interesting question for the $k$-median problem (or any other center based clustering problem): consider the potentially infinite-sized LP which has variables $y_i$ for all
 $i\in X$, but restrict to the optimal solution which only is allowed to open centers from $C$. How big is this ``integrality gap''? It is not too hard to show that for the $k$-median problem this is between $2$ and $4$. The upper bound gives hope we can get a true $4$-approximation for the continuous $k$-median problem, but it seems one would need new ideas to obtain such a result.


\paragraph*{Organization of this Paper}
In the main body, we focus on the continuous $\lambda$-UFL and the continuous fair $k$-median results, since we believe that they showcase the technical ideas in this paper. Proofs of certain statements have been deferred to the appendix.
The description of the results on continuous $k$-means and continuous $k$-center with outliers can be found in~\Cref{appsec:kmeans} and~\Cref{appsec:kcwo}, respectively.

\section{Preliminaries}\label{sec:prelim}
Given a metric space $(X,d)$ on points $X$ with pairwise distances $d$, we use the notation $d(v,S) = \min_{i \in S}d(v,i)$ for $v \in X$ and $S\subseteq X$ to denote $v$'s distance to the set $S$.
\begin{definition}[Continuous $k$-median (\contkmed)]
The input is a metric space $(X,d)$, clients $C \subseteq X, \abs{C} = n \in \NN$, and $k \in \NN$. The goal is to find $S \subseteq X, \abs{S} = k$ minimizing $\cost(S) := \sum_{v \in C}d(v,S)$.
\end{definition}
\begin{definition}[Continuous Fair $k$-median (\contfairkmed)]
Given the \contkmed input, plus fairness radii $r : C \to [0,\infty)$, the goal is to find $S \subseteq X, \abs{S} = k$ such that $\forall v \in C$, $d(v,S) \leq r(v)$, minimizing $\cost(S)$.
\end{definition}




In the Uncapacitated Facility Location (UFL) problem, the restriction of opening only $k$ facilities is replaced by having a cost associated with opening each facility. When these costs are equal to the same value $\lambda$ for all facilities, the problem is called $\lambda$-UFL.
\begin{definition}[Continuous $\lambda$-UFL (\conlamufl)]
Given a metric space $(X,d)$, clients $C \subseteq X, \abs{C} = n \in \NN$, and $\lambda \geq 0$, find $S \subseteq X$ that minimizes the sum of ``connection cost'' $\costc(S) := \sum_{v \in C} d(v,S)$ and ``facility opening cost'' $\costf(S) := \lambda|S|$.
\end{definition}
Let $\Delta = \max_{u,v \in C} d(u,v)$ denote the diameter of a metric $(X,d)$. For $x \in X$, $0 \leq r \leq \Delta$, the {\em ball of radius $r$ around $x$} is $B(x,r) := \{x' \in X \mid d(x',x) \leq r\}$. Throughout the paper, we use balls of the form $B(v,r)$ where $v$ is a client and $r \in \RR$. To circumvent the potentially infinite number of radii, the radii can be discretized into $\Ie = \{\eps, 2\eps, \dots, \ceil{\Delta/\eps}\eps\}$ for a small constant $\eps = O(1/n^2)$. Thereupon, we can appeal to the following lemma to bound the size of $\Ie$ by $O(n^5)$.
\begin{lemma}[Rewording of Lemma 4.1, \cite{Ahmadian2017BetterGF}]\label{lem:delta-vs-n}
Losing a factor of $\paren{1+\frac{100}{n^2}}$, we can assume that for any $u,v \in C$, $1 \leq d(u,v) \leq n^3$.
\end{lemma}
For simplicity of exposition, we present our techniques using radii in $\RR$, and observe that discretizing to $\Ie$ incurs an additive loss of at most $O(n\eps) = O(1/n)$ in our guarantees. We also note that $\log\opt \leq \log(n\Delta) = O(\log n)$ by the above, which enables us to efficiently binary-search over our guesses $\opt_g$.

\section{Continuous \texorpdfstring{$\lambda$}{lambda}-UFL}\label{sec:ufl}
We start this section with our $2.32$-approximation for \conlamufl (\Cref{thm:ufl-approx}). For this, we introduce a new linear programming formulation, and adapt the rounding algorithm of Shmoys-Tardos-Aardal to the new program. The resulting procedure exhibits our main ideas, and serves as a warm-up for the remaining sections. Also, in \Cref{sec:ufl:hardness}, we prove that it is \NP-hard to approximate \conlamufl within a factor of $2-o(1)$, using ideas due to Cohen-Addad, Karthik, and Lee~\cite{cohen-addad-et-al}. This shows that the continuous version cannot be approximated as well as the discrete version, which has a best-known approximation factor of 1.463~\cite{Li2013}.

\subsection{Approximation algorithm}\label{sec:ufl:approx}
This subsection is dedicated to proving the following theorem:
\begin{theorem}\label{thm:ufl-approx}
There is a polynomial time algorithm that, for an instance of \conlamufl with optimum $\opt$, yields a solution with cost at most $(\frac{2}{1-e^{-2}} + \eps)\opt < 2.32\cdot \opt$. Here $\eps = O(\frac{1}{n^2})$.
\end{theorem}
We design the following linear program for \conlamufl. We use variables $C_v$ for the connection cost of each client $v$, and $y(v,r)$ for the number of facilities opened within each ball of the form $B(v,r)$. We also use a guess of the optimum $\opt_g$, which we will soon discuss how to obtain. Throughout, we use $y(B)$ as shorthand for $y(v,r)$ where $B=B(v,r)$. 
\begin{align}
\lambda \sum_{B \in \B}y(B) + \sum_{v \in C}C_v &\leq \opt_g & \forall \B \subseteq  \{B(v,r)\}_{\substack{v \in C \\ r \in \RR}} \text{ pairwise disjoint} \tag{\sf UFL}\label{ufllp:cost}\\
\int_0^\infty \paren{1-y(v,r)}\dee r &\leq C_v & \forall v \in C, r \in \RR\tag{{\sf UFL}-1}\label{ufllp:cov}\\
y(v,r) &\leq y(v,r')& \forall v \in C,\;r,r' \in \RR\text{ s.t. }r \leq r'\tag{{\sf UFL}-2}\label{ufllp:mon}\\
y(v,r) \geq 0,\ & C_v \geq 0& \forall v \in C, r \in \RR\notag
\end{align}
Observe that, given a solution $S \subseteq X$ of cost at most $\opt_g$, we can obtain a feasible solution of \ref{ufllp:cost} as follows. For client $v \in C$, we set $C_v = d(v,S)$. For $v \in C, r \in \RR$, we set $y(v,r) = 0$ for $r < d(v,S)$ and $y(v,r) = 1$ for $r \geq d(v,S)$.

Our approach is to round a solution $(C,y)$ of \ref{ufllp:cost}. Observe that there are polynomially many constraints of the form  \eqref{ufllp:cov} and \eqref{ufllp:mon}; hence, we can efficiently obtain a solution $(C,y)$ that satisfies them. So for the remainder of this section, we assume that those constraints are satisfied. On the other hand, there are infinitely many constraints of type \eqref{ufllp:cost}. This is why we employ a round-or-cut framework via the ellipsoid algorithm~\cite{ellipsoid}. We begin with an arbitrary $\opt_g$, and when ellipsoid asks us if a proposed solution $(C,y)$ is feasible, we run the following algorithm.

The algorithm inputs $\alpha < 1$, and defines $\ra(v)$ as the minimum radius at which client $v \in C$ has at least $\alpha$ mass of open facilities around it. First, all clients are deemed {\em uncovered} $(U = C)$. Iteratively, the algorithm picks the $j$, i.e the uncovered client, with the smallest $\ra(j)$. $j$ is put into the set $\RepsCya$. Any client $v$ within distance $\ra(j) + \ra(v)$ of $j$ is considered a $\child$ of $j$ and is now {\em covered}. When all clients are covered, i.e. $U = \emptyset$, the algorithm outputs $\RepsCya$.

\begin{algorithm}[H]\caption{Filtering for \conlamufl}\label{alg:ufl-reps}
\begin{algorithmic}[1]
\Require A proposed solution $(\{C_v\}_{v \in C},\{y(v,r)\}_{v \in C, r \in \RR})$ for \ref{ufllp:cost}, parameter $\alpha \in (e^{-2},1)$
\State $\ra(v) \gets \min\{r \in \RR\mid y(v,r) \geq \alpha\}$ for all $v \in C$
\State $\RepsCya \gets \emptyset$ \Comment{``representative'' clients}
\State $U \gets C$ \Comment{``uncovered'' clients}
\While{$U \neq \emptyset$}
	\State \label{alg:ufl:ln:rep} Pick $j \in U$ with minimum $\ra(j)$
	\State \label{alg:ufl:ln:child} $\child(j) \gets \{v \in U \mid d(v,j) \leq \ra(v) + \ra(j)\}$
	\State $U \gets U \setminus \child(j)$
	\State $\RepsCya \gets \RepsCya + j$
\EndWhile
\Ensure $\RepsCya$
\end{algorithmic}
\end{algorithm}

Notice that, by construction, the collection of balls $\braces{B(j,\ra(j))}_{j \in \RepsCya}$ is pairwise disjoint. Hence, the following constraint, which we call $\SepCya$, is of the form \eqref{ufllp:cost}:
\begin{align}
    \lambda \sum_{j \in \RepsCya}y(j,\ra(j)) + \sum_{v \in C}C_v \leq \opt_g\tag{$\SepCya$}\label{SepCya}
\end{align}
We will show that
\begin{lemma}\label{lma:uflcost}
If $(C,y)$ satisfies \eqref{ufllp:cov}, \eqref{ufllp:mon}, and \ref{SepCya}, then there exists a suitable $\alpha \in (e^{-2},1)$ for which the output of \Cref{alg:ufl-reps} has cost at most $\frac 2 {1-e^{-2}}\opt_g$.
\end{lemma}

Thus, if we find that the desired approximation ratio is not attained, then it must be that \ref{SepCya} was not satisfied, and we can pass it to ellipsoid as a separating hyperplane. If ellipsoid finds that the feasible region of our linear program is empty, then we increase $\opt_g$ and try again. Otherwise, we obtain a solution $\RepsCya$ that attains the desired guarantees.




We now analyze \Cref{alg:ufl-reps} to prove \Cref{lma:uflcost}.
\begin{proof}[Proof of \Cref{lma:uflcost}]For this proof, we will fix $(C,y)$, and refer to $\RepsCya$ as $\Repsa$.

To prove a suitable $\alpha$ exists, assume $\alpha$ is picked uniformly at random from $(\beta,1)$ for some $0 < \beta < 1$; we will see later that $\beta = e^{-2}$ is optimal. Take $\Repsa$, the output of \Cref{alg:ufl-reps} on $(C,y)$. By definition of $\ra$, $\sum_{j \in \Repsa}y(j,\ra(j)) \geq \alpha\abs\Repsa$. Thus $\costf(\Repsa) = \lambda \abs{\Repsa} \leq \frac{1}{\alpha}\cdot \lambda \sum_{j \in \Reps} y(j,\ra(j))$, which implies
\begin{equation}\label{eq:uflrepfaccost}
    \Exp[\costf(\Repsa)] \leq \frac{\ln(1/\beta)}{(1-\beta)} \lambda \sum_{j \in \Repsa} y(j,\ra(j)).
\end{equation}

To bound the expected connection cost, take $v \in C$ and observe that, since all the clients are ultimately covered in \Cref{alg:ufl-reps}, there has to exist $j \in \Repsa$ for which $v \in \child(j)$. By construction of $\child$, $d(v,j) \leq \ra(v) + \ra(j)$, which is at most $2\ra(v)$ by our choice of $j$ in Line~\ref{alg:ufl:ln:rep}. Thus, for any client $v$, we get $d(v,\Repsa) \leq d(v,j) \leq 2\ra(v)$. So we are left to bound $\Exp\bracket{\ra(v)}$ for an arbitrary client $v \in C$.

We have that $\Exp\bracket{\ra(v)} = \frac 1 {1-\beta} \int_\beta^1 \ra(v)\dee\alpha
\leq \frac 1 {1-\beta} \int_0^1 \ra(v)\dee\alpha$. We notice that at $\alpha = y(v,r), \ra(v) = r$. Also, $r_0(v) = 0$. So given \eqref{ufllp:mon} for all balls $B(v,r)$ with $r \in \RR$, we can apply a change of variable to the integral to get $\int_0^1 \ra(v)\dee \alpha = \int_{r_0(v)}^{r_1(v)}\paren{y(v,r_1(v))-y(v,r)}\dee r \leq \int_0^\infty\paren{1-y(v,r)}\dee r\leq C_v$,
where the last inequality is by \eqref{ufllp:cov}. Thus we have $\Exp\bracket{\ra(v)} \leq \frac {C_v} {1-\beta}$. Summing $d(v, \Repsa)$ over all $v \in C$ we have
\begin{equation}\label{eq:uflrepconcost}
    \Exp[\costc(\Repsa)] = \sum_{v \in C} \Exp[d(v,\Repsa)] \leq 2\sum_{v \in C} \Exp[\ra(v)] \leq \frac 2 {1-\beta}\sum_{v \in C}C_v.
\end{equation}

To balance $\costf$ from \eqref{eq:uflrepfaccost} and $\costc$ from \eqref{eq:uflrepconcost}, we set $\beta = e^{-2}$. The expected \conlamufl cost of $\Repsa$ is, using \ref{SepCya},
\begin{equation*}
    \Exp[\costf(\Repsa) + \costc(\Repsa)] \leq \frac 2 {1-e^{-2}}\paren{\lambda \sum_{j \in \Repsa} y(j,\ra(j)) + \sum_{v \in C}C_v} \leq \frac {2\cdot\opt_g} {1-e^{-2}}.
\end{equation*}
Since the bound holds in expectation over a random $\alpha$, there must exist an $\alpha \in (\beta,1)$ that satisfies it deterministically.
\end{proof}
To obtain a suitable $\alpha$, we can adapt the derandomization procedure from the discrete version~\cite{shmoys-tardos-aardal-1997}. The procedure relies on having polynomially many interesting radii; for this, we recall that while we have used $r \in \RR$ for simplicity, our radii are actually $r \in \Ie$, $\abs\Ie = O(n^5)$.

\subsection{Hardness of approximation}\label{sec:ufl:hardness}

Our hardness result for this problem is as follows:
\begin{restatable}{theorem}{uflhard}\label{thm:ufl-hard}
Given an instance of \conlamufl and $\eps > 0$, it is $\NP$-hard to distinguish between the following:
\begin{itemize}
    \item There exists $S \subseteq X$ such that $\costf(S) + \costc(S) \leq (1+6\eps)n$
    \item For any $S \subseteq X$, $\costf(S) + \costc(S) \geq (2-\eps)n$
\end{itemize}
\end{restatable}

Thus we exhibit hardness of approximation up to a factor of $\frac{2-\eps}{1+6\eps}$, which tends to $2$ as $\eps \to 0$. Our reduction closely follows the hardness proof for \contkmed \cite{cohen-addad-et-al}. We relegate the details to \Cref{appsec:ufl}.
\section{Continuous Fair \texorpdfstring{$k$}{k}-Median}\label{sec:fairness}
The main result of this section is the following theorem.
\begin{theorem}
\label{thm:main}
There exists a polynomial time algorithm for \contfairkmed that, for an instance with optimum cost $\opt$, yields a solution with cost at most $8\opt+\eps$, in which, each client $v \in C$ is provided an open facility within distance $8r(v)+\eps$ of itself. Here $\eps = O\paren{\frac 1 {n^2}}$.
\end{theorem}


We create a round-or-cut framework, via the ellipsoid algorithm~\cite{ellipsoid}, that adapts the Chakrabarty-Negahbani algorithm \cite{chakrabarty-negahbani-fairness} to the continuous setting. For this, we will modify the \ref{ufllp:cost} linear program to suit \contfairkmed. As before, $\opt_g$ is a guessed optimum, $C_v$ is the cost share of a client $v \in C$, and $y(v,r)$ represents the number of facilities opened in $B(v,r)$. There are two key modifications. First, we expand the monotonicity constraints of the form \eqref{ufllp:mon} to include {\em non-concentric} balls, which are crucial for adapting the fairness guarantee of Chakrabarty and Negahbani \cite{chakrabarty-negahbani-fairness}. Second, we enforce the fairness constraints by requiring $y(v,r(v)) \geq 1$ for each client $v \in C$.

\begin{align}
\sum_{v \in C}C_v &\leq \opt_g \tag{\lp}\label{fairlp:cost}\\
\sum_{B \in \B}y(B) &\leq k  &\forall \B \subseteq  \{B(v,r)\}_{\substack{v \in C \\ r \in \RR}} \text{ pairwise disjoint} \tag{\lp-1}\label{fairlp:k}\\
\int_0^\infty\paren{1-y(v,r)}\dee r &\leq C_v & \forall v \in C\tag{\lp-2}\label{fairlp:markovplus}\\
y(u,r) &\leq y(v,r')& \forall u,v \in C, r,r' \in \RR, B(u,r) \subseteq B(v,r') \tag{\lp-3}\label{fairlp:mon}\\
y(v,r(v)) &\geq 1 &\forall v \in C \tag{\lp-4}\label{fairlp:fair}\\
y(v,r) \geq 0,\ & C_v \geq 0& \forall v \in C, r \in \RR\notag
\end{align}


We will frequently use the following property of \ref{fairlp:cost}. See \Cref{appsec:lp} for the proof.
\begin{restatable}{lemma}{markovlp}\label{lma:markovplus-to-markov}
Consider a solution $(C,y)$ of \ref{fairlp:cost}. If for a client $v$, $(C,y)$ satisfies all constraints of the form \eqref{fairlp:markovplus} and \eqref{fairlp:mon} involving $v$, then for any $r_0 \in \RR$, $C_v \geq r_0\paren{1-y(v,r_0)}$.
\end{restatable}

As before, we will only worry about the constraints that are exponentially many. These are \eqref{fairlp:k}. For this, we use ellipsoid~\cite{ellipsoid}. Given a proposed solution $(C,y)$ of \ref{fairlp:cost}, we construct $\RepsCy \subseteq C$, as follows.

We first perform a filtering step. For each $v \in C$, we define $R(v):=\min{r(v),2C_v}$.  In the beginning, all clients are ``uncovered'' (i.e. $U = C$). In each iteration, let $j \in U$ be the uncovered client with the minimum $R(j)$; and add $j$ to our set of ``representatives'' $\RepsCy$. Any $v \in U$ within distance $2R(v)$ of $j$ (including $j$ itself) will be added to the set $\child(j)$, and will be removed from $U$. After all clients are covered, i.e. $U = \emptyset$, the algorithm outputs $\RepsCy$. For a formal description of this algorithm, see \Cref{appsec:reps}.

For a $j \in \RepsCy$, let $s(j)$ be the closest client to $j$ in $\RepsCy\setminus \set{j}$. Let $a(j) := \frac{d(j,s(j))}2$. So the collection of balls $\braces{B(j,a(j))}_{j \in \RepsCy}$ is pairwise disjoint, and the following constraint, which we call $\SepCy$, is of the form \eqref{fairlp:k}.
\begin{align}
    \sum_{j \in \RepsCy}y\paren{j,a(j)} \leq k \tag{$\SepCy$}\label{SepCy}
\end{align}
We have that
\begin{lemma}\label{lma:half-mass}
If $(C,y)$ satisfies \eqref{fairlp:markovplus}-\eqref{fairlp:fair} and \ref{SepCy}, then
\begin{enumerate}[ref = \thetheorem.\arabic*]
    \item $\forall j \in \RepsCy,\;y\paren{j,a(j)} \geq \half$\label{lma:fairkmed:yball-half}
    \item $\abs{\RepsCy} \leq 2k$\label{lma:fairkmed:Reps-2k}
\end{enumerate}
\end{lemma}
\begin{proof}
Fix $j \in \RepsCy$. By Line~\ref{alg:fair-reps:ln:child} of \Cref{alg:fair-reps},
\begin{align}
a(j) = \frac{d(j,s(j))}2 \geq R(j)\label{eq:djsj-Rj}
\end{align}
So if $R(j) = r(j)$, then by \eqref{fairlp:mon} and \eqref{fairlp:fair}, $y\paren{j,a(j)} \geq y(j,r(j)) \geq 1$. Else $R(j) = 2C_j$. By \Cref{lma:markovplus-to-markov}, $C_j \geq a(j)\paren{1-y\paren{j,a(j)}}$.

If $C_j = 0$, then this implies $y\paren{j,a(j)}\geq 1$. 
Otherwise, substituting $a(j)$ by $R(j)$ from \eqref{eq:djsj-Rj} and setting $R(j) = 2C_j$ gives $C_j \geq 2C_j\paren{1-y\paren{j,a(j)}}$, i.e. $y\paren{j,a(j)} \geq \half$. Now, by \ref{SepCy}, we have $k \geq \sum_{j \in \RepsCy}y\paren{j,a(j)} \geq \half \cdot \abs{\RepsCy}$.
\end{proof}

So if we find that $\abs\RepsCy > 2k$, then \ref{SepCy} must be violated, and we can pass it to ellipsoid as a separating hyperplane. Hence in polynomial time, we either find that our feasible region is empty, or we get $(C,y)$ and $\RepsCy$ such that $(C,y)$ satisfies \eqref{fairlp:cost}, \eqref{fairlp:markovplus}-\eqref{fairlp:fair}, and \ref{SepCy}. In the first case, we increase $\opt_g$ and try again. In the latter case, we round $(C,y)$ further to attain our desired approximation ratios, via a rounding algorithm that we will now describe. This algorithm focuses on $\RepsCy$ and ignores other clients, as justified by the following lemma.

\begin{restatable}{lemma}{reps}
\label{lma:fair:reps}
$S \subseteq X$ be a solution to \contfairkmed. Consider a proposed solution $(C,y)$ of \ref{fairlp:cost} that satisfies \eqref{fairlp:cost}. Then $\sum_{v \in C}d(v,S) \leq \sum_{j \in \RepsCy}\abs{\child(j)}d(j,S) + 4\opt_g$.
\end{restatable}
The proof closely follows from a standard technique for the discrete version~\cite{cgst-primal-rounding,chakrabarty-negahbani-fairness}. We provide the proof in \Cref{appsec:fairness:reps}.

Our algorithm will also ignore facilities outside $\RepsCy$, so our solution will be a subset of $\RepsCy$. For the remainder of this section, we fix $(C,y)$, and refer to $\RepsCy$ as $\Reps$. We write the following polynomial-sized linear program, \ref{lpfairkmed}, where $\Reps$ are the only clients and the only facilities. The objective function of \ref{lpfairkmed} is a lower bound on $\sum_{j \in \Reps}\abs{\child(j)}d(j,S)$, so hereafter we compare our output with \ref{lpfairkmed}. We do not include fairness constraints in this program, and we will see later that it is not necessary to do so.

In \ref{lpfairkmed}, the variables $z_i$ for each $i \in \Reps$ denote whether $i$ is open as a facility. The variables $x_{ij}$ for $i,j \in \Reps$ denote whether the client $j$ uses the facility $i$. 

\begin{align}
    \textrm{minimize} \sum_{j \in \Reps}\abs{\child(j)}&\sum_{i \in \Reps}x_{ij}d(j,i) \tag{\DLP} \label{lpfairkmed}\\
    \sum_{i \in \Reps}z_i &\leq k\tag{\DLP-1}\label{lpfairkmed:k}\\
    \sum_{i \in \Reps}x_{ij} &= 1 &\forall j \in \Reps\tag{\DLP-2}\label{lpfairkmed:cov}\\
    x_{ij} &\leq y_i &\forall i,j \in \Reps\tag{\DLP-3}\label{lpfairkmed:sanity}\\
    x_{ij} \geq 0,& \;z_i \geq 0&\forall i,j \in \Reps\notag
\end{align}

We will now round $(C,y)$ to an integral solution of \ref{lpfairkmed}. Our first step is to convert $(C,y)$ to a fractional solution $(\bar x, \bar z)$ of \ref{lpfairkmed}. To do this, for each $j \in \Reps$, we consolidate the $y$-mass in $B(j,a(j))$ onto $j$, i.e. we set $\bar z_j = y(j,a(j))$. By Lemma~\ref{lma:fairkmed:yball-half}, each $\bar z_j$ is then at least $\half$. This allows $j$ to use only itself and $s(j)$ as its fractional facilities.

\begin{algorithm}[H]\caption{Consolidation for \contfairkmed\label{alg:fair:consolidate}}
    \begin{algorithmic}[1]
        \Require A proposed solution $\paren{\braces{C_v}_{v \in C}, \braces{y(v,r)}_{v \in C, r \in \Ie}}$ for \ref{fairlp:cost}, and $\Reps$ from \Cref{alg:fair-reps}
        \For{$j \in \Reps$}
        \State $s(j) \gets \arg\min_{v \in \Reps\setminus j} d(j,v)$
        \State $a(j) \gets d(j,s(j))/2$
        \State $\bar z_j \gets \min\braces{y\paren{j,a(j)},1}$
        \State $\bar x_{jj} \gets \bar z_j$
        \State $\bar x_{js(j)} \gets 1-\bar z_j$
        \EndFor
        \Ensure $(\bar x, \bar z)$
    \end{algorithmic}
\end{algorithm}

\begin{lemma}\label{lma:fair:consolidate}
$\forall j \in \Reps$, $\bar z_j \geq \half$, and $(\bar x, \bar z)$ is a feasible solution of \ref{lpfairkmed} with cost at most $2\opt_g$.
\end{lemma}
\begin{proof}
For a $j \in \Reps$, if $y\paren{j,a(j)} = 1$ then $\bar z_j = 1$. Otherwise, by Lemma~\ref{lma:fairkmed:yball-half}, $\bar z_j = y\paren{j,a(j)} \geq \half$.

Hence $\bar 1-\bar z_j \leq \half \leq \bar z_{s(j)}$, which implies feasibility by construction and \ref{SepCy}. It also implies that $\sum_{i \in \Reps}\bar x_{ij}d(j,i) = (1-\bar z_j)d(j,s(j))$. If $y\paren{j,a(j)} = 1$, then the RHS above is $0 \leq 2C_j$. Otherwise $\sum_{i \in \Reps}\bar x_{ij}d(j,i) = \paren{1-y\paren{j,a(j)}}d(j,s(j)) = 2\paren{1-y\paren{j,a(j)}}a(j)\leq 2C_j$
where the last inequality follows from \Cref{lma:markovplus-to-markov}. Multiplying by $\abs{\child(j)}$ and summing over all $j \in \Reps$, we have by Line~\ref{alg:fair-reps:ln:rep} in \Cref{alg:fair-reps},
\begin{align*}
    \sum_{j \in \Reps}\abs{\child(j)}\sum_{i \in \Reps}\bar x_{ij}d(j,i)\leq 2\sum_{j \in \Reps}\abs{\child(j)}C_j &\leq 2\sum_{j \in \Reps}\sum_{v \in \child(j)}C_v = 2\sum_{v \in C}C_v
\end{align*}
which is at most $2\opt_g$ by \eqref{fairlp:cost}.
\end{proof}
Now, to round $(\bar x, \bar z)$ to an integral solution, we appeal to an existing technique \cite{cgst-primal-rounding,chakrabarty-negahbani-fairness}. We state the relevant result here, and provide the proof in \Cref{appsec:fairness:rounding}.

\begin{restatable}[\cite{cgst-primal-rounding, chakrabarty-negahbani-fairness}]{lemma}{cgst}\label{lma:fair:use-cgst}
Let $(\bar x, \bar z)$ be a feasible solution of \ref{lpfairkmed} with cost at most $2\opt_g$, such that $\forall j \in \Reps$, $\bar z_j \geq \half$. Then there exists a polynomial time algorithm that produces $S \subseteq \Reps$ such that
\begin{enumerate*}[label = {\color{darkgray} \textsf{\textbf{(\arabic*)}}}, ref = \thetheorem.\arabic*]
    \item $\abs S = k$;\label{lma:fair:use-cgst:k}
    \item If $\bar z_j = 1$, then $j \in S$;\label{lma:fair:use-cgst:1s}
    \item $\forall j \in \Reps$, at least one of $j,s(j)$ is in $S$; and\label{lma:fair:use-cgst:j-or-sj}
    \item $\displaystyle \sum_{j \in \Reps}\abs{\child(j)}d(j,S) \leq 4\opt_g$.\label{lma:fair:use-cgst:cost}
\end{enumerate*}
\end{restatable}

Thus, by \Cref{lma:fair:reps}, we have shown that $\sum_{v \in C}d(v,S) \leq \sum_{j \in \Reps}\abs{\child(j)}\sum_{i \in S}d(j,S) +4\opt_g \leq 8\opt_g$. Now we show the fairness ratio, adapting a related result~\cite[Lemma 3]{chakrabarty-negahbani-fairness} from the discrete version. This is where we crucially require the monotonicity constraints \eqref{fairlp:mon} for {\em non-concentric} balls.

\begin{lemma}\label{lma:fair:fair}
$\forall v \in C$, $d(v,S) \leq 8r(v)$.
\end{lemma}
\begin{proof}
Fix $v \in C$, and let $j \in \Reps$ such that $v \in \child(j)$.  By construction of $\child(j)$ in \Cref{alg:fair-reps},
\begin{align}
    d(v,j) \leq 2R(v) \leq 2r(v) \label{eq:fair-close}
\end{align}

So if $j \in S$, then we are done. Otherwise, by Lemma~\ref{lma:fair:use-cgst:1s}, $\bar z_j < 1$, i.e. by \Cref{alg:fair:consolidate}, $y\paren{j,a(j)} < 1$. But by the fairness constraints \eqref{fairlp:fair}, $y(v,r(v)) \geq 1$. So by the monotonicity constraints \eqref{fairlp:mon}, $B(v,r(v)) \not\subseteq B(j,a(j))$, as otherwise we would have $1 \leq y(v,r(v)) \leq y(j,a(j)) < 1$, a contradiction.

So fix $w \in B(v,r(v)) \setminus B\paren{j,a(j)}$. We have
\begin{align}
    a(j) &< d(j,w),\quad d(v,w) \leq r(v)\label{eq:choice-of-w}
\end{align}

By Lemma~\ref{lma:fair:use-cgst:j-or-sj}, either $j \in S$ or $s(j) \in S$, so
\begin{align*}
    d(j,S) &\leq d(j,s(j)) = 2a(j) < 2d(j,w)&\dots\text{by \eqref{eq:choice-of-w}}\\
    &\leq 2\paren{d(v,j) + d(v,w)} \leq 2\paren{2r(v)+r(v)} &\dots\text{by \eqref{eq:fair-close} and \eqref{eq:choice-of-w}}\\
    &= 6r(v)
\end{align*}
So, by \eqref{eq:fair-close}, $d(v,S) \leq d(v,j) + d(j,S) \leq 2r(v) + 6r(v) = 8r(v)$.
\end{proof}
Thus we have proved \Cref{thm:main}.

We observe here that, by the simple reduction of setting all $r(v)$'s to $\infty$, \Cref{thm:main} implies a solution of cost $8\opt + \eps$ for \contkmed. We improve this ratio via an improved rounding procedure by Charikar, Guha, Shmoys, and Tardos \cite{cgst-primal-rounding}, which rounds $(\bar x, \bar z)$ such that $\sum_{j \in \Reps}\abs{\child(j)}d(j,S) \leq \frac 8 3 \opt_g$, instead of the $4\opt_g$ that we obtain above. This yields:
\begin{corollary}\label{cor:unfair:better}
There exists a polynomial time algorithm for \contkmed that, on an instance with optimum cost $\opt$, yields a solution of cost at most $6 \frac 2 3\opt + \eps$.
\end{corollary}
This improved rounding, however, no longer guarantees to open either $j$ or $s(j)$ for each $j \in \Reps$. Such a guarantee (Lemma~\ref{lma:fair:use-cgst:j-or-sj}) is crucial to our fairness bound in \Cref{lma:fair:fair}. So the improvement is not naively adaptable to \contfairkmed.
\section*{Acknowledgements}
We thank the reviewers for their comments, which helped us improve our exposition and exhibit a stronger connection to existing work.

\clearpage
\bibliographystyle{plain}
\bibliography{references}
\appendix
\section{Details of the \texorpdfstring{\contfairkmed}{Cont-fair-k-Med} results}\label{appsec:fairness}
\subsection{Filtering algorithm}\label{appsec:reps}
\begin{algorithm}[H]\caption{Filtering for \contfairkmed}\label{alg:fair-reps}
\begin{algorithmic}[1]
\Require A proposed solution $\paren{\{C_v\}_{v \in C},\{y(v,r)\}_{v \in C, r \in \RR}}$ for \ref{fairlp:cost}
\State $R(v) \gets \min\{r(v),2C_v\}$ for all $v \in C$
\State $\RepsCy \gets \emptyset$ \Comment{``representative'' clients}
\State $U \gets C$ \Comment{``uncovered'' clients}
\While{$U \neq \emptyset$}
	\State \label{alg:fair-reps:ln:rep} Pick $j \in U$ with minimum $R(j)$
	\State \label{alg:fair-reps:ln:child} $\child(j) \gets \{v \in U \mid d(v,j) \leq 2R(v)\}$
	\State $U \gets U \setminus \child(j)$
	\State $\RepsCy \gets \RepsCy \cup \set{j}$
\EndWhile
\Ensure $\RepsCy$
\end{algorithmic}
\end{algorithm}

\subsection{Proof of \texorpdfstring{\Cref{lma:markovplus-to-markov}}{Markov Lemma}}\label{appsec:lp}

\markovlp*
\begin{proof}
\begin{align*}
    C_v &\geq \int_0^\infty \paren{1-y(v,r)}\dee r &\dots\text{\eqref{fairlp:markovplus}}\\
    &\geq \int_0^{r_0} \paren{1-y(v,r)}\dee r\\
    &\geq \paren{1-y(v,r_0)}\int_0^{r_0}\dee r&\dots\text{by \eqref{fairlp:mon}, }y(v,r) \leq y(v,r_0)\\
    &=r_0\paren{1-y(v,r_0)}
\end{align*}
\end{proof}

\subsection{Proof of \texorpdfstring{\Cref{lma:fair:reps}}{Reps Lemma}}\label{appsec:fairness:reps}

\reps*
\begin{proof}
Consider a client $v \in C$. In \Cref{alg:fair-reps}, all clients are eventually covered, so there must exist $j \in \RepsCy$ such that $v \in \child(j)$. Thus 
\[d(v,S) \leq d(j,S) + d(v,j) \leq d(j,S) + 2R(v)\]
where the last inequality follows by Line~\ref{alg:fair-reps:ln:child}. Summing this over all $v \in C$,
\begin{align*}
    \sum_{v \in C}d(v,S) &\leq \sum_{j \in \RepsCy}\sum_{v \in \child(j)}\paren{d(j,S) + 2R(v)}\\
    &\leq \sum_{j \in \RepsCy}\abs{\child(j)}d(j,S) + 4\sum_{v \in C}C_v
\end{align*}
where the last inequality follows because $R(v) \leq 2C_v$. The lemma follows by \eqref{fairlp:cost}.
\end{proof}

\subsection{Proof of \texorpdfstring{\Cref{lma:fair:use-cgst}}{Rounding Lemma}}\label{appsec:fairness:rounding}
\cgst*
We first round $(\bar x, \bar z)$ to a $\braces{\half,1}$-integral solution $(\tilde x, \tilde z)$ of the same cost. Then, we round $(\tilde x, \tilde z)$ to an integral solution $(x,z)$ with twice the cost.
\subsubsection*{Rounding to a $\braces{\half,1}$-integral solution}
We first exhibit this procedure for \Cref{lma:fair:use-cgst}, and then modify it slightly for \Cref{lma:kmeans:use-cgst}. We will assume that $\sum_{j \in \Reps}\bar z_j$ is an integer because, if it is not, then it is strictly less than $k$, and we can distribute the value $k - \sum_{j \in \Reps}\bar z_j$ among the $\bar z$-values of some arbitrary representatives, without increasing the quantity $\sum_{i \in X}\bar x_{ij}d(j,i)$ for any $j \in \Reps$. Hence
\begin{assumption}\label{assm:integral-z}
	$\sum_{j \in \Reps}\bar z_j \in \ZZ$.
\end{assumption}
Given the assumption, we round to a $\braces{\half, 1}$-solution as follows: we start with $\tilde z$ being $\bar z$, and then we find two representatives $j_1,j_2$ that violate $\braces{\half,1}$-integrality. We will pick the one with a lower value of $\abs{\child(j)}d(j,s(j))$; say this is $j_1$. Then we will decrease $\tilde z_{j_1}$ and increase $\tilde z_{j_2}$ until one of them hits $\half$ or $1$. We will repeat this until $\tilde z$ becomes $\braces{\half,1}$-integral. Alongside, $\tilde x$ values are adjusted accordingly.

\begin{algorithm}[H]\caption{Rounding to a $\braces{\half,1}$-integral solution \label{alg:half-integral}}
    \begin{algorithmic}[1]
    \Require A solution $\paren{\braces{\bar x_{ij}}_{i,j \in \Reps}, \braces{\bar z_j}_{j \in \Reps}}$ of \ref{lpfairkmed}; other $\tilde x, \tilde z$ values implicitly zero
    \State $(\tilde x, \tilde z) \gets (\bar x, \bar z)$ \label{alg:half-integral:ln:init}
    \While{$\exists j_1,j_2 \in \Reps, j_1 \neq j_2, \half < \tilde z_{j_1}, \tilde z_{j_2} < 1$}\label{alg:half-integral:ln:whilecond}
        \State \Comment{wlog assume $\abs{\child(j_1)}d(j_1,s(j_1)) \leq \abs{\child(j_2)d(j_2,s(j_2))}$}
        \State $\delta(j_1,j_2) \gets \min\braces{\tilde z_{j_2}-\half, 1-\tilde z_{j_1}}$
        \State $\tilde z_{j_1} \gets \tilde z_{j_1} - \delta(j_1,j_2)$
        \State $\tilde z_{j_2} \gets \tilde z_{j_2} + \delta(j_1,j_2)$
        \For{$j \in \braces{j_1,j_2}$}
            \State $\tilde x_{jj} \gets \tilde z_j$
            \State $\tilde x_{s(j)j} \gets \tilde 1-\tilde z_j$
        \EndFor
    \EndWhile
    \Ensure $(\tilde x, \tilde z)$
    \end{algorithmic}
\end{algorithm}

\begin{lemma}[\cite{chakrabarty-negahbani-fairness, cgst-primal-rounding}]\label{lma:half-integral}
Let $(\bar x, \bar z)$ be a feasible solution of \ref{lpfairkmed} with cost at most $2\opt_g$, such that $\forall j \in \Reps, \bar z_j \geq \half$. Then $(\tilde x, \tilde z)$ is a feasible $\braces{\half,1}$-integral solution of \ref{lpfairkmed} of cost at most $2\opt_g$, such that $\forall j \in \Reps, \bar z_j = 1 \implies \tilde z_j = 1$.
\end{lemma}

\begin{proof}

Observe that \Cref{alg:half-integral} never alters $\tilde z$-values that are $0$ or $1$ in Line~\ref{alg:half-integral:ln:init}. Moreover, it never lowers a $\tilde z$-value that is at least $\half$ to a value strictly below $\half$. So the property
\[\forall j \in \Reps, \tilde z_j \geq \half\]
as initially ensured, is preserved throughout. Thus, our construction never violates \eqref{lpfairkmed:sanity}; and, by \Cref{assm:integral-z}, we can find $j_1,j_2$ as in Line~\ref{alg:half-integral:ln:whilecond} iff $\tilde z$ is not $\braces{\half,1}$-integral. 

Now consider the quantity
\[\psi(\tilde x, \tilde z) := \sum_{j \in \Reps}\max\braces{\tilde z_j - \half, 1-\tilde z_j}\]
Observe that the while-loop in \Cref{alg:half-integral} strictly decreases $\psi(\tilde x, \tilde z)$ in every iteration, but never makes it negative. Also, $\psi(\tilde x, \tilde z) = 0$ iff $(\tilde x, \tilde z)$ is $\braces{\half,1}$-integral. So the while-loop terminates and yields a $\braces{\half,1}$-integral solution. Note that every iteration of the while-loop preserves $\sum_{j \in \Reps}\tilde z_j$. So we preserve \eqref{lpfairkmed:k}. We satisfy \eqref{lpfairkmed:cov} by construction.

Finally, for the cost bound:
\[\sum_{j \in \Reps}\abs{\child(j)}\sum_{i \in X}\tilde x_{ij}d(j,i) = \sum_{j \in \Reps}\tilde x_{s(j)j}d(j,s(j))\]
In one iteration of the while-loop, $\tilde x_{s(j_1)}$ decreases by $\delta(j_1,j_2)$, while $\tilde x_{s(j_2)}$ increases by the same. So the RHS above, in one iteration, increases by
\[\delta(j_1,j_2)\paren{\abs{\child(j_1)}d(j_1,s(j_1))-\abs{\child(j_2)}d(j_2,s(j_2))}\]
which is designed to be non-positive in Line~\ref{alg:half-integral:ln:whilecond}. Thus by Line~\ref{alg:half-integral:ln:init} and \Cref{lma:fair:consolidate} we have
\[\sum_{j \in \Reps}\abs{\child(j)}\sum_{i \in X}\tilde x_{ij}d(j,i) \leq \sum_{j \in \Reps}\abs{\child(j)}\sum_{i \in X}\bar x_{ij}d(j,i) \leq 2\opt_g\]
\end{proof}
\subsubsection*{Rounding to an integral solution}

Given $(\tilde x, \tilde z)$ as constructed above, we first open all $j \in \Reps$ that have $\tilde z_j = 1$. We discard these representatives, and also any representatives that now have $\tilde z_{s(j)} = 1$. On the remaining $\Reps$, we construct a forest, as follows: we draw directed edges of the form $(j,s(j))$, and arbitrarily discard any antiparallel edges that occur when $s(s(j)) = j$. Then we open alternate levels of every tree in this forest, choosing either odd or even levels to get the lower cost.

\begin{algorithm}[H]\caption{Rounding to an integral solution\label{alg:integral}}
    \begin{algorithmic}[1]
    \Require A $\braces{\half, 1}$-integral solution $\paren{\braces{\tilde x_{ij}}_{i,j \in \Reps}, \braces{\tilde z_j}_{j \in \Reps}}$ of \ref{lpfairkmed}, other $\tilde x, \tilde z$ implicitly zero
        \State $(x,z) \gets$ all zeros
        \State $V \gets \Reps$ \Comment{vertex set}
        \For{$j \in \Reps$ such that $\tilde z_j = 1$}
            \State $z_j \gets 1$, $x_{jj} \gets 1$
            \State $V \gets V \setminus j$
        \EndFor
        \For{$j \in \Reps$ such that $z_j = 1$}
            \State $x_{s(j)j} \gets 1$
            \State $V \gets V\setminus j$
        \EndFor\label{alg:integral:line:V}
        \State $E \gets \braces{(j,s(j)) \mid j \in V}$ \Comment{edge set, has antiparallel edges}
        \For{Each connected component $T$ in $(V,E)$} 
            \State Pick $j \in T$ such that $s(s(j)) = j$ \Comment{unique, corresponds to closest pair in $T$}
            \State $E \gets E \setminus (s(j),j))$ \Comment{removing antiparallel edges}
            \State $\treeroot(T) \gets s(j)$ \Comment{now all edges point upwards in the tree, and levels well-defined}
            \State $T_\odd \gets$ odd level vertices in $T$
            \State $T_\even \gets$ even level vertices in $T$
            \If{$\sum_{j \in T}d(j,T_\odd) \leq \sum_{j \in T}d(j,T_\even)$}
                \State $\forall j \in T_\odd, z_j \gets 1, x_{jj} \gets 1$
                \State $\forall j \in T_\even, x_{s(j)j} \gets 1$
            \Else
                \State $\forall j \in T_\even, z_j \gets 1, x_{jj} \gets 1$
                \State $\forall j \in T_\odd, x_{s(j)j} \gets 1$
            \EndIf
        \EndFor
        \Ensure $(x,z)$
    \end{algorithmic}
\end{algorithm}

\begin{lemma}[\cite{cgst-primal-rounding}] \label{lma:integral}
	Let $(\tilde x, \tilde z)$ be a feasible $\braces{\half,1}$-integral solution of \ref{lpfairkmed} with cost at most $2\opt_g$. Then $(x,z)$ is a feasible integral solution of \ref{lpfairkmed} with cost at most $4\opt_g$, such that $\forall j \in \Reps$, $\tilde z_j = 1 \implies z_j = 1$; and either $z_j = 1$ or $z_{s(j)} = 1$.
\end{lemma}

\begin{proof}
By construction, $(x,z)$ is integral, and satisfies \eqref{lpfairkmed:sanity} and \eqref{lpfairkmed:cov}. The construction also ensures that for each $j \in \Reps$, at least one of $z_j, z_{s(j)}$ is set to $1$.

Now we show that $(x,z)$ satisifies \eqref{lpfairkmed:k}. Let $\Reps_\half = \braces{j \in \Reps \mid \tilde z_j = \half}$. \Cref{alg:integral} sets all $z$ values in $\Reps \setminus \Reps_\half$ to $1$, and sets at most $\frac{\abs{\Reps_\half}}{2}$ of the $z$ values in $\Reps_\half$ to $1$. By \Cref{assm:integral-z}, and since \Cref{alg:half-integral} maintains $\sum_{i \in X}\tilde z_i = \sum_{i \in X}\bar z_i$, we have that $\abs {\Reps_\half}$ is even, so $\ceil{\frac{\abs{\Reps_\half}}{2}} = \frac{\abs{\Reps_\half}}{2}$. Thus
		\[\sum_{i \in X}z_i = \abs{\Reps \setminus \Reps_\half} + \frac{\abs{\Reps_\half}}{2} = \sum_{j \in \Reps}\tilde z_j\]
which is at most $k$ since $(\tilde x, \tilde z)$ is feasible.
	
Finally, for the cost bound, it suffices to show the following: for each $j \in \Reps$,
\[\sum_{i \in X}x_{ij}d(j,i) \leq 2\sum_{i \in X}\tilde x_{ij}d(j,i).\]
If $z_j = \tilde z_j = 1$, then LHS = RHS = $0$ above. Otherwise $\tilde z_j = \half = \tilde x_{s(j)j}$. So the RHS is exactly $2 \cdot \half \cdot d(j,s(j)$, and since either $z_j = 1$ or $z_{s(j)} = 1$, the LHS is at most $d(j,s(j))$.
\end{proof}

Setting $S := \braces{j \in \Reps \mid z_j = 1}$ yields \Cref{lma:fair:use-cgst}.

\section{\texorpdfstring{$32$}{32}-approximation for \texorpdfstring{\contkmeans}{Cont-k-Means}}\label{appsec:kmeans}
\begin{definition}[Continuous $k$-means (\contkmeans)]
The input is a metric space $(X,d)$, clients $C \subseteq X, \abs{C} = n \in \NN$, and $k \in \NN$. The goal is to find $S \subseteq X, \abs{S} = k$ minimizing $\cost(S) := \sum_{v \in C}d(v,S)^2$.
\end{definition}

\begin{theorem}\label{thm:kmeans}
    There exists a polynomial time algorithm that, for an instance of \contkmeans with optimum $\opt$, yields a solution of cost at most $32\opt + \eps$. Here $\eps = O\paren{\frac 1 {n^2}}$.
\end{theorem}
To prove this, we proceed with a round-or-cut framework similar to the one in \Cref{sec:fairness} using a linear program similar to \ref{fairlp:cost} (except the fairness constraints). Here the variables are $\set{C_v^2}_{v \in C}$ for the cost shares and $\set{y(v,r)}_{v \in C, r \in \RR}$ for the facility masses. The constraints involving only the $y(v,r)$'s remain the same as in \ref{fairlp:cost}, while those that involve $C_v$'s are now adapted to the squared costs.

\begin{align}
\sum_{v \in C}C_v^2 &\leq \opt_g \tag{$\lp^2$}\label{kmeanslp:cost}\\
\sum_{B \in \B}y(B) &\leq k  &\forall \B \subseteq  \{B(v,r)\}_{\substack{v \in C \\ r \in \RR}} \text{ pairwise disjoint} \tag{$\lp^2$-1}\label{kmeanslp:k}\\
2\int_0^\infty\paren{1-y(v,r)}r\dee r &\leq C_v^2 & \forall v \in C\tag{$\lp^2$-2}\label{kmeanslp:markovplus}\\
y(u,r) &\leq y(v,r')& \forall u,v \in C, r,r' \in \RR, B(u,r) \subseteq B(v,r') \tag{$\lp^2$-3}\label{kmeanslp:mon}\\
y(v,r) \geq 0,\ & C_v^2 \geq 0& \forall v \in C, r \in \RR\notag
\end{align}
The rounding algorithm and its analysis are analogous to our results in \Cref{sec:fairness}. We first prove the following analogue of \Cref{lma:markovplus-to-markov} for \ref{kmeanslp:cost}.
\begin{lemma}\label{lma:markovplus-to-markov-kmeans}
Consider a solution $(C^2,y)$ of \ref{kmeanslp:cost}. If for a client $v$, $(C^2,y)$ satisfies all constraints of the form \eqref{kmeanslp:markovplus} and \eqref{kmeanslp:mon} involving $v$, then for any $r_0 \in \RR$, $C_v^2 \geq r_0^2\paren{1-y(v,r_0)}$.
\end{lemma}
\begin{proof}
\begin{align*}
    C_v^2 &\geq 2\int_0^\infty \paren{1-y(v,r)}r\dee r &\dots\text{\eqref{kmeanslp:markovplus}}\\
    &\geq 2\int_0^{r_0} \paren{1-y(v,r)}r\dee r\\
    &\geq \paren{1-y(v,r_0)}\cdot 2\int_0^{r_0}r\dee r&\dots\text{by \eqref{kmeanslp:mon}, }y(v,r) \leq y(v,r_0)\\
    &=r_0^2\paren{1-y(v,r_0)}
\end{align*}
\end{proof}

The filtering procedure is the same as \Cref{alg:fair-reps}, except $R(v)$ now equals $\sqrt 2 C_v$. Also, for clarity, we call the output $\RepsCCy$. For a $j \in \RepsCCy$, $s(j)$ and $a(j)$ are defined as before, and we use the same separating constraint, which we now call \ref{SepCCy}.
\begin{align}
    \sum_{j \in \RepsCCy}y(j,a(j)) \leq k\tag{$\SepCCy$}\label{SepCCy}
\end{align}
As before, we will focus on satisfying \ref{SepCCy} and argue that this is sufficient for our rounding algorithm. We prove an analogue of \Cref{lma:half-mass}.

\begin{lemma}
If $(C^2,y)$ satisfies \eqref{kmeanslp:markovplus}-\eqref{kmeanslp:mon} and \ref{SepCCy}, then
\begin{enumerate}[ref = \thetheorem.\arabic*]
    \item $\forall j \in \RepsCCy,\;y\paren{j,a(j)} \geq \half$\label{lma:kmeans:yball-half}
    \item $\abs{\RepsCCy} \leq 2k$\label{lma:kmeans:Reps-2k}
\end{enumerate}
\end{lemma}
\begin{proof}
Fix $j \in \RepsCCy$. By construction of $\child$ sets, we have
\begin{align}
a(j) = \frac{d(j,s(j))}2 \geq R(v) = \sqrt 2 C_j\label{eq:djsj-Rj-kmeans}
\end{align}
By \Cref{lma:markovplus-to-markov-kmeans},
\begin{align*}
    C_j^2 &\geq a(j)^2\paren{1-y\paren{j,a(j)}}
\end{align*}
If $C_j = 0$, since $a(j) > 0$ by definition, this implies $y\paren{j,a(j)}\geq 1$. 
Otherwise, substituting $a(j)$ by $R(v) = \sqrt 2 C_j$ from \eqref{eq:djsj-Rj-kmeans} gives
\begin{align*}
    C_j^2 &\geq 2C_j^2\paren{1-y\paren{j,a(j)}} \implies  y\paren{j,a(j)} \geq \half.
\end{align*}
Now, by \ref{SepCCy},
\[k \geq \sum_{j \in \RepsCCy}y\paren{j,a(j)} \geq \half \cdot \abs{\RepsCCy}\]
\end{proof}

So if we find that $\abs\RepsCCy > 2k$, we can pass \ref{SepCCy} to ellipsoid as a separating hyperplane. Hence for our rounding algorithm, we can now use a $(C^2,y)$ that satisfies \ref{SepCCy}. We fix such a $(C^2,y)$, and hereafter refer and $\RepsCCy$ as simply $\Reps$. We now prove an analog of \Cref{lma:fair:reps}, which allows us to ignore clients outside of $\Reps$ hereafter.

\begin{lemma}\label{lma:kmeans:reps}
$S \subseteq X$ be a solution to \contkmeans. Consider a proposed solution $(C^2,y)$ of \ref{kmeanslp:cost} that satisfies \eqref{kmeanslp:cost}. Then
\begin{align*}
    \sum_{v \in C}d(v,S)^2 \leq 2\sum_{j \in \RepsCy}\abs{\child(j)}d(j,S)^2 + 16\opt_g
\end{align*}
\end{lemma}
\begin{proof}
Consider a client $v \in C$. In the filtering algorithm, all clients are eventually covered, so there must exist $j \in \RepsCCy$ such that $v \in \child(j)$. Thus 
\[d(v,S) \leq d(j,S) + d(v,j) \leq d(j,S) + 2R(v) = d(j,S) + 2\sqrt 2C_v\]
where the last inequality follows by Line~\ref{alg:fair-reps:ln:child}. Squaring, we get
\[d(v,S)^2 \leq (d(j,S) + 2\sqrt 2C_v)^2 \leq 2(d(j,S)^2 + (2\sqrt 2C_v)^2) = 2(d(j,S)^2 + 8C_v^2)\]
Summing this over all $v \in C$,
\begin{align*}
    \sum_{v \in C}d(v,S)^2 &\leq 2\sum_{j \in \RepsCy}\sum_{v \in \child(j)}\paren{d(j,S)^2 + 8C_v^2}\\
    &\leq 2\sum_{j \in \RepsCy}\abs{\child(j)}d(j,S)^2 + 16\sum_{v \in C}C_v^2
\end{align*}
The lemma follows by \eqref{kmeanslp:cost}.
\end{proof}
Next, following our algorithm for \contfairkmed closely, we will round $(C^2,y)$ to an integral solution of \eqref{lpkmeans}, which has the same constraints as \ref{lpfairkmed}, but the objective function changes to suit the $k$-means problem:
\begin{align}
    \textrm{minimize} \sum_{j \in \Reps}\abs{\child(j)}&\sum_{i \in \Reps}x_{ij}d(j,i)^2 \tag{$\DLP^2$} \label{lpkmeans}
\end{align}

This is achieved via the same consolidation procedure as in \Cref{alg:fair:consolidate}. We prove an analogue of \Cref{lma:fair:consolidate}.

\begin{lemma}\label{lma:kmeans:consolidate}
$(\bar x, \bar z)$ is a feasible solution of \ref{lpkmeans} with cost at most $4\opt_g$ and, $\forall j \in \Reps$, $\bar z_j \geq \half$.
\end{lemma}
\begin{proof}
For a $j \in \Reps$, if $y\paren{j,a(j)} = 1$ then $\bar z_j = 1$. Otherwise, by Lemma~\ref{lma:kmeans:yball-half}, $\bar z_j = y\paren{j,a(j)} \geq \half$.

Hence $\bar 1-\bar z_j \leq \half \leq \bar z_{s(j)}$, which implies feasibility by construction and $\Sep$. We also have
\[\sum_{i \in \Reps}\bar x_{ij}d(j,i)^2 = (1-\bar z_j)d(j,s(j))^2\]
If $y\paren{j,a(j)} = 1$, then the RHS above is $0 \leq 4C_j^2$. Otherwise
\[\sum_{i \in \Reps}\bar x_{ij}d(j,i)^2 = \paren{1-y\paren{j,a(j)}}d(j,s(j))^2 = 4\paren{1-y\paren{j,a(j)}}a(j)^2\leq 4C_j^2\]
where the last inequality follows from \Cref{lma:markovplus-to-markov-kmeans}. Multiplying by $\abs{\child(j)}$ and summing over all $j \in \Reps$, we have by construction of $\Reps$,
\begin{align*}
    \sum_{j \in \Reps}\abs{\child(j)}\sum_{i \in \Reps}\bar x_{ij}d(j,i)^2 \leq 4\sum_{j \in \Reps}\abs{\child(j)}C_j^2 &\leq 4\sum_{j \in \Reps}\sum_{v \in \child(j)}C_v^2\\
    &= 4\sum_{v \in C}C_v^2 \leq 4\opt_g
\end{align*}
\end{proof}

Finally, we prove an analog of \Cref{lma:fair:use-cgst}.

\begin{lemma}[Modified from \cite{cgst-primal-rounding, chakrabarty-negahbani-fairness}]\label{lma:kmeans:use-cgst}
Let $(\bar x, \bar z)$ be a feasible solution of \ref{lpkmeans} with cost at most $4\opt_g$, such that $\forall j \in \Reps$, $\bar z_j \geq \half$. Then there exists a polynomial time algorithm that produces $S \subseteq \Reps$ such that
\begin{enumerate}[ref = \thetheorem.\arabic*]
    \item $\abs S = k$\label{lma:kmeans:use-cgst:k}
    \item If $\bar z_j = 1$, then $j \in S$.\label{lma:kmeans:use-cgst:1s}
    \item $\forall j \in \Reps$, at least one of $j,s(j)$ is in $S$\label{lma:kmeans:use-cgst:j-or-sj}
    \item $\displaystyle \sum_{j \in \Reps}\abs{\child(j)}d(j,S)^2 \leq 8\opt_g$\label{lma:kmeans:use-cgst:cost}
\end{enumerate}
\end{lemma}
\begin{proof}
The proof is similar to that of \Cref{lma:fair:use-cgst}. We first apply \Cref{alg:half-integral}, except we now take $j_1, j_2$ such that $\abs{\child(j_1)}d(j_1,s(j_1))^2 \leq \abs{\child(j_2)}d(j_2,s(j_2))$. This, via similar reasoning as the proof of \Cref{lma:half-integral}, gives us a $\set{\half,1}$-integral solution $(\tilde x, \tilde z)$ of \ref{lpkmeans} such that $\forall j \in \Reps, \bar z_j = 1 \implies \tilde z_j = 1$, and
\[\sum_{j \in \Reps}\abs{\child(j)}\sum_{i \in X}\tilde x_{ij}d(j,i)^2 \leq \sum_{j \in \Reps}\abs{\child(j)}\sum_{i \in X}\bar x_{ij}d(j,i)^2 \leq 4\opt_g\]
We then apply \Cref{alg:half-integral}. Similarly to the proof of \Cref{lma:integral}, we get an integral solution $(x,z)$ satisfying the necessary \ref{lpkmeans} constraints. We also have that $\forall j \in \Reps, \tilde z_j = 1 \implies z_j = 1$ and either $z_j = 1$ or $z_{s(j)} = 1$. So we need to show that
\[\sum_{i \in X}x_{ij}d(j,i)^2 \leq 2\sum_{i \in X}d(j,i)^2\]
If $z_j = \tilde z_j = 1$, then both the LHS and the RHS are $0$. Otherwise $\tilde z_j = \half = \tilde x_{s(j)j}$. So the RHS is exactly $2\cdot\half\cdot d(j,s(j))^2$, and the LHS is at most $d(j,s(j))^2$.
\end{proof}

So by \Cref{lma:kmeans:reps}, we have, for the $S$ constructed above,
\[\sum_{v \in C}d(v,S)^2 \leq 2 \cdot 8\opt_g + 16\opt_g = 32\opt_g\]

\subsubsection*{Generalizing further to Continuous \texorpdfstring{$(k,p)$}{(k,p)}-clustering}
Our approach extends to the generalization where the objective function is the sum of some $p^{th}$ power of the distances.
\begin{definition}[Continuous ($k,p$)-clustering (\contkp)]
The input is a metric space $(X,d)$, clients $C \subseteq X, \abs{C} = n \in \NN$, and $k,p \in \NN$. The goal is to find $S \subseteq X, \abs{S} = k$ minimizing $\cost(S) := \sum_{v \in C}d(v,S)^p$.
\end{definition}

\begin{theorem}\label{thm:kp}
    There exists a polynomial time algorithm that, for an instance of \contkp with optimum $\opt$, yields a solution of cost at most $2^{2p+1}\opt + \eps$. Here $\eps = O\paren{\frac 1 {n^2}}$.
\end{theorem}

This follows by the same framework as for \contkmed and \contkmeans, with the following alterations:
\begin{description}
    \item[LP formulation.] The cost-share variables are $\set{C_v^p}_{v \in C}$. \eqref{kmeanslp:cost} is replaced by $\sum_{v \in C}C_v^p \leq \opt_g$, and \eqref{kmeanslp:markovplus} is replaced by $p\int_0^1 (1-y(v,r))r^{p-1}\dee r \leq C_v^p$.
    \item[Filtering.] $R(v) := 2^{1/p}C_v$.
    \item[Linear program on $\Reps$.] The objective function in \eqref{lpkmeans} is replaced by \[\sum_{j \in \Reps}\abs{\child(j)}\sum_{i \in \Reps}x_{ij}d(j,i)^p.\]
    \item[Rounding to $\set{\half,1}$-integral solution.] In \Cref{alg:half-integral}, we use $j_1,j_2$ such that \[\abs{\child(j_1)}d(j_1,s(j_1))^p \leq \abs{\child(j_2)}d(j_2,s(j_2))^p.\]
\end{description}

The guarantees are obtained using the relaxed triangle inequality $(a+b)^p \leq 2^{p-1}(a^p+b^p)$. Via the filtering step, we get
\[\sum_{v \in C}d(v,S)^2 \leq 2^{p-1}\sum_{j \in S}\abs{\child(j)}d(j,S) + 2^{2p}\opt_g\]
and the consolidation step incurs a factor of $2^p$. The remaining algorithm incurs a factor of $2$ irrespective of the value of $p$. So we get an approximation factor of $2^{p-1}\cdot 2^p \cdot 2 + 2^{2p} = 2^{2p+1}$.

\section{Continuous \texorpdfstring{$k$}{k}-center with Outliers}\label{appsec:kcwo}
In the $k${\sf-center} problem, the input is the same as that for $k${\sf-median}, but the objective is now the {\em maximum} distance that a client has to the open facilities, rather than the sum. Furthermore, we consider a version here that allows {\em outliers}, i.e. only some $m$ clients need to be considered when computing the cost.
\begin{definition}[Continuous $k$-center with Outliers (\contkcwo)]
The input is a metric space $(X,d)$, clients $C \subseteq X, \abs{C} = n \in \NN$, and $k,m \in \NN$. The goal is to find $S \subseteq X, \abs{S} = k$ and $D \subseteq C, \abs D \geq m$ minimizing $\max_{v \in D}d(v,S)$.
\end{definition}

In this section, we provide an algorithm for \contkcwo that matches the state of the art for discrete \kcwo up to a small additive factor. The best approximation algorithm for the discrete version is a 2-approximation primal-rounding approach~\cite{ChakrGK2020}. Later, Chakrabarty and Negahbani \cite{chakrabarty-negahbani-f-center} extended this result to allow for generalized constraints on facilities, using a round-or-cut framework. Both of those results are tight due to a hardness of approximation by Hochbaum and Shmoys~\cite{hochbaum-shmoys}. Via our idea of creating $y$-variables for balls around clients, we adapt the discrete results~\cite{chakrabarty-negahbani-fairness, ChakrGK2020} to the \contkcwo problem.

\begin{theorem}\label{thm:kcwo}
There exists a polynomial time algorithm that, for an instance of \contkcwo with optimum $\opt$, yields a solution of cost at most $2\opt+\eps$. Here $\eps = O\paren{\frac 1 {n^2})}$.
\end{theorem}

As before, we proceed with a round-or-cut framework with a guessed optimum $\opt_g$, and variables $y(v,r) = y(B)$ representing the number of open facilities in each ball $B=B(v,r)$. We will also have variables $\cov_v$ for each $v \in C$ representing whether or not $v$ has a facility open within distance $\opt_g$ of it. Our linear program \kcwo is as follows:

\begin{align}
    \sum_{B \in \B}y(B) &\leq k &\forall \B=\braces{B(v,r)}_{\substack{v \in C \\ r \in \RR}}\text{ pairwise disjoint}&\tag{\kcwo-1}\label{kcwolp:k}\\
    y(v,\opt_g) &\geq \cov_v &\forall v \in C \tag{\kcwo-2}\label{kcwolp:sanity}\\
    y(v,r) &\leq y(v,r') &\forall v \in C, r,r' \in \RR \tag{{\sf kCwO}-3}\label{kcwolp:mon}\\
    \sum_{v \in C}\cov_v &\geq m \tag{{\sf kCwO}-4}\label{kcwolp:m}
\end{align}
All constraints except for \eqref{kcwolp:k} are polynomially many, and we use ellipsoid for \eqref{kcwolp:k}. Given a proposed solution $(\cov,y)$ of \kcwo, we construct $\Repscov \subseteq C$ and $\Scov \subseteq \Repscov$ as follows. We first produce $\Repscov$ such that all clients are within $2\opt_g$ of it, but it is potentially of size larger than $k$. Then, we pick $\Scov \subseteq \Repscov, \abs \Scov = k$ by choosing the $k$ largest $\child$ sets.

Constructing $\Repscov$ is simple: initially all clients are uncovered ($U = C$), and we pick the uncovered client $i \in U$ with maximum $\cov_i$. We put $i$ in $\Repscov$, and consider all uncovered clients within distance $2\opt_g$ of $i$ to be in $\child(i)$, which are then covered. We repeat this until all clients are covered.

\begin{algorithm}[H]\caption{Rounding for \contkcwo\label{alg:kcwo}}
\begin{algorithmic}[1]
    \Require $\braces{\cov_v}_{v \in C}, \opt_g$
    \State $\Repscov \gets \emptyset$
    \State $U \gets C$
    \While{$U \neq \emptyset$}
        \State Pick $i \in U$ with maximum $\cov_i$ \label{alg:kcwo:ln:child}
        \State $\Repscov \gets \Repscov + i$
        \State $\child(i) \gets U \cap B(i,2\opt_g)$
        \State $U \gets U \setminus B(i,2\opt_g)$
    \EndWhile
    \State $\Scov \gets \emptyset$
    \While{$\abs{\Scov} < k$}\Comment{construct $\Scov$ to be the $k$ points that cover the most}
        \State Pick $i \in \Repscov\setminus \Scov$ with maximum $\abs{\child(i)}$
        \State $\Scov \gets \Scov+i$
    \EndWhile
    \Ensure $\Scov$
\end{algorithmic}
\end{algorithm}

Observe that, by construction of $\Repscov$ in \Cref{alg:kcwo}, the collection of balls $\B_\cov := \braces{B(i,\opt_g)}_{i \in \Repscov}$ is pairwise disjoint. So the following constraint, called \ref{Sepcov} is of the form \eqref{kcwolp:k}:
\begin{align}
    \sum_{i \in \Repscov}y(i,\opt_g) \leq k\tag{$\Sepcov$}\label{Sepcov}
\end{align}

We will show that
\begin{lemma}\label{lma:kcwo}
If $(\cov, y)$ satisfies \eqref{kcwolp:sanity}-\eqref{kcwolp:m} and \ref{Sepcov}, then
\[\sum_{i \in \Scov}\abs{\child(i)} \geq m\]
\end{lemma}
So if we get $\sum_{i \in \Scov}\abs{\child(i)} < m$, we know that \ref{Sepcov} must be violated, so we return it to ellipsoid as a separating hyperplane. Hence, in polynomial time, we can either conclude that our guess of $\opt_g$ should be larger, or obtain $\Scov$ such that $\sum_{i \in \Scov}\abs{\child(i)} \geq m$. Since $\child(i)$ sets are pairwise disjoint by construction, putting $D = \sqcup_{i \in \Scov}\abs{\child(i)}$ yields the $2$-approximation.

\begin{proof}[Proof of \Cref{lma:kcwo}] For a fixed $\opt_g$ and $(\cov, y)$, let $\Reps = \Repscov, S = \Scov$. By \eqref{kcwolp:sanity} and \ref{Sepcov},
\begin{align}
    \sum_{i \in \Reps}\cov_i \leq \sum_{i \in \Reps}y(i,\opt_g) \leq k \label{eq:covs}
\end{align}
Also, by \Cref{alg:kcwo} and \eqref{kcwolp:m},
\begin{align}
    \sum_{i \in \Reps}\abs{\child(i)}\cov_i &\geq \sum_{i \in \Reps}\sum_{v \in \child(i)}\cov_v = \sum_{v \in C}\cov_v \geq m &\dots\text{by \eqref{kcwolp:m}} \label{eq:bcovs}
\end{align}

By \eqref{eq:covs} and \eqref{eq:bcovs}, we can view $\sum_{i \in \Reps}\abs{\child(i)}\cdot\frac{\cov_i}k$ as an average of quantities $\abs{\child(i)}$ with weights $\frac{\cov_i}k$. Since this average is at least $m/k$, the $k$ largest $\abs{\child(i)}$'s must sum up to at least $m$, as desired.
\end{proof}
\section{Hardness of approximation for \texorpdfstring{\conlamufl}{Cont-UFL}}\label{appsec:ufl}

\uflhard*

We will reduce from the following result, that was used \cite{cohen-addad-et-al} to show a hardness of $(2-\eps)$ for \contkmed.
\begin{theorem}[\cite{KMS17,DKKMS18a,DKKMS18b,BKS18,KMS18}]Given a graph $G=(V,E)$ and an $\eps > 0$, it is $\NP$-hard to distinguish between the following:
\begin{itemize}
    \item There exists pairwise disjoint independent sets $V_1,V_2,V_3,V_4 \subseteq V$ such that $\forall i \in [4]$, $\abs{V_i}\geq \frac{1-\eps}4\abs V$.
    \item There is no independent set in $G$ of size $\eps\abs V$.
\end{itemize}
\end{theorem}

\begin{proof}[Proof of \Cref{thm:ufl-hard}]
Given a graph $G=(V,E)$ and $\eps > 0$, we construct the following instance of the \conlamufl problem.

For $m := \abs E$, let $X = \RR^m$, and let $d$ be the distance induced in $\RR^m$ by the $\ell_\infty$-norm. Fix an arbitrary orientation on $E$, and let the clients be $C := \braces{A(v) \mid v \in V}$ where, for a $v \in V$ and $e=(u',v') \in E$,
    \[A(v)(e) := \begin{cases}
    2 & \text{if }v=u'\\
    -2 & \text{if }v=v'\\
    0 & \text{otherwise}
    \end{cases}\]
So we can hereafter say $n = \abs V = \abs C$. Now let $\lambda = \eps n$. Also, fix an arbitrary optimal solution of this \conlamufl instance that opens $r$ facilities for some $r \geq 1$, and set $\eps' = \frac \eps r$.

The proof follows by \Cref{lma:ufl:reduction}, which analyzes the two directions of this reduction.
\end{proof}

\begin{lemma}\label{lma:ufl:reduction}
Consider a \conlamufl instance constructed from the graph $G=(V,E)$, as per the construction described earlier. Suppose the optimum in this instance opens $r$ facilities, and let $\eps' = \eps/r$.

If there are pairwise disjoint independent sets $V_1,V_2,V_3,V_4 \subseteq V$, each of size at least $\frac{1-\eps'}4 n$, then there is a solution in the \conlamufl instance with cost at most $(1+6\eps)n$. Conversely, if there is no independent set of size $\eps' n$ in $G$, then the optimum in the \conlamufl instance is at least $(2-\eps)n$.
\end{lemma}

\begin{proof}
\textbf{Large independent sets imply low-cost UFL:} Given pairwise disjoint independent sets $V_1,V_2,V_3,V_4 \subseteq V$, each of size $\frac{1-\eps'}4n \geq \frac{1-\eps}4n$ in $G$, we can produce a solution in the \conlamufl instance that has cost at most $(1+6\eps)n$. To construct this solution, open the facilities $s_1,s_2,s_3,s_4$, defined as follows. For $i \in [4]$ and each coordinate $e = (u',v') \in E$,
\[s_i(e) := \begin{cases}
1 & \text{if }u' \in V_i\\
-1 & \text{if }v' \in V_i\\
0 &\text{otherwise}
\end{cases}\]
Since coordinates in any $A(v)$ are in $\braces{\pm 2,0}$ and coordinates in each $s_i$ are in $\braces{\pm 1, 0}$, we get the simple upper bound of $d(A(v),s_i) \leq 3$ for any $i \in [4], v \in V$.

For some $i \in [4], v \in V_i$ and $e=(u',v') \in E$, we study the quantity $\abs{A(v)(e) - s_i(e)}$. If $u' = v$, then it is $\abs{2-1} = 1$. If $v' = v$, then it is $\abs{-2-(-1)} = 1$. If $u' \in V_i-v$, then it is $\abs{0-1} = 1$. If $v' \in V_i-v$, then it is $\abs{0-(-1)} = 1$. When $e$ lies outside $V_i$, then it is zero. So when $v \in V_i$, we get a tighter bound of $d(A(v),s_i) \leq 1$.
Thus the cost of connecting clients to facilities is at most
\begin{align*}
    \sum_{i=1}^4\abs{V_i} + 3\abs{\paren{V \setminus \bigcup_{i=1}^4 V_i}} \leq 4\cdot \frac{1-\eps}4n + 3\eps n \leq (1+2\eps)n
\end{align*}
We open $4$ facilities, so the facility opening cost is $4\lambda = 4\eps n$. Thus the total cost of this solution is $(1+6\eps)n$, as promised.

\noindent \textbf{Small maximum independent set implies high-cost UFL:} Suppose there is no independent set in $G$ of size $\eps'n$. In the optimum solution of \conlamufl that we had fixed, let the facilities opened be $S = \braces{s_1,\dots,s_r}$ and, for $i \in [r]$, let the set $C_i$ be the clients assigned to $s_i$ in this solution in other words, $d(v,S) = d(v,s_i)$ for all $i \in [r],v \in C_i$. Also let $V_i := \braces{v \in V \mid A(v) \in C_i}$. We have the following claim (Claim 1, \cite{cohen-addad-et-al}):
\begin{claim} \label{claim:large-matching}
If $G$ has no independent set of size $\eps'n$, then the induced subgraph of $G$ on $V_i$ contains a matching $M_i$ of size at least $\frac{\abs{V_i}-\eps'n} 2$.
\end{claim}
We prove the claim at the end of this section. Given the claim, we have that
\begin{align*}
    \sum_{v \in V_i} d(A(v),s_i) &\geq \sum_{(u,v) \in M_i}\paren{d(A(v),s_i)+d(A(u),s_i}\\
    &\geq \sum_{(u,v) \in M_i}d(A(u),A(v))\\
    &=\sum_{(u,v) \in M_i}\abs{A(u)((u,v))-A(v)((u,v))}\\
    &=4\abs{M_i}\\
    &\geq 2\paren{\abs{V_i}-\eps'n}
\end{align*}
and thus $\costc(S)$ is
\begin{equation*}
    \sum_{j \in C}d(j,S) = \sum_{i=1}^r\sum_{v \in V_i}d(A(v),s_i)
    \geq 2n - 2\eps'rn
    = 2n-2\eps n
\end{equation*}
Since $r$ facilities are open, $\costf(S)$ is $\lambda r = r\eps n \geq \eps n$. Thus $\costf(S) + \costc(S)$ is at least $n(2-2\eps+\eps) \geq n(2-\eps)$.
\end{proof}

\begin{proof}[Proof of \Cref{claim:large-matching} \cite{cohen-addad-et-al}]
If $\abs{V_i} \geq \eps' n$, then it cannot be independent. So pick an edge in the subgraph induced by $V_i$, include it in the matching $M_i$, and discard its endpoints from $V_i$. Repeat this process until $\abs{V_i} < \eps' n$ to obtain the desired $M_i$.
\end{proof}
\end{document}